\documentclass[journal,9 pt]{IEEEtran}

\IEEEoverridecommandlockouts

\usepackage{color}
\usepackage{enumerate,graphicx,epstopdf}
\usepackage{amsmath,amssymb,bm,mathtools}
\usepackage{cite,subfig}
\usepackage[overload]{empheq}
\usepackage{cleveref,array} 
\usepackage{mathtools}
\usepackage{amsthm}

\newtheorem{ass}{Assumption}
\newtheorem{rmk}{Remark}
\newtheorem{prb}{Problem}

\newtheorem{trm}{Theorem}
\newtheorem{prp}{Proposition}
\newtheorem{cor}{Corollary}

\newcommand{\real}{\mathbb{R}}
\newcommand{\unit}{\mathbb{U}}
\newcommand{\quat}{\mathbb{Q}}
\newcommand{\sqr}{$\hfill\square$}

\title{Spacecraft Attitude Control with Nonconvex Constraints:\\ An Explicit Reference Governor Approach \thanks{The authors are with the University of Michigan, Ann Arbor. Email: \texttt{\{mnicotra,dliaomcp,lburlion,ilya\}@umich.edu}.  This research is supported by the National Science Foundation Award Number CMMI 1562209.}}
\author{Marco M. Nicotra, Dominic Liao-McPherson, Laurent Burlion, Ilya V. Kolmanovsky}

\begin{document}
\maketitle

\begin{abstract}
This paper introduces a novel attitude controller for spacecraft subject to actuator saturation and multiple exclusion cone constraints. The proposed solution relies on a two-layer approach where the first layer prestabilizes the system dynamics whereas the second layer enforces constraint satisfaction by suitably manipulating the reference of the prestabilized system. In particular, constraint satisfaction is guaranteed by taking advantage of set invariance properties, whereas asymptotic convergence is achieved by implementing a non-conservative navigation field which is devoid of undesired stagnation points. Multiple numerical examples illustrate the good behavior of the proposed scheme.
\end{abstract}

\section{Introduction}
Spacecraft must often perform large angle reorientation maneuvers to accomplish their missions (e.g. to point a camera, antenna, or telescope at a different celestial or terrestrial objects). In many cases, these maneuvers are complicated by the fact that spacecraft often carry sensitive instruments which cannot be (or must be) pointed in certain directions \cite{koenig2009novel,hablani1999attitude}. A notable example is the Cassini spacecraft where a constraint monitoring module was used to ensure that certain sensors were never pointed towards the sun\cite{singh1997constraint}. 

Most existing solutions for constrained reorientation with exclusion zones can be broadly divided into path planning methods, trajectory optimization methods, potential field methods, and anti-windup methods. Broadly speaking, path planning methods function by discretizing the attitude space and generating a feasible path by applying a graph search type method \cite{kjellberg2015discretized,tanygin2016fast,weiss2014spacecraft} or suitable random search algorithms \cite{cui2007onboard,frazzoli2001randomized}. Random and graph search algorithms can be computationally expensive, and may neglect dynamics, but often have probabilistic completeness or asymptotic optimality guarantees. 


Trajectory optimization methods, which include model predictive control \cite{eren2017model} approaches, pose the constrained reorientation problem as a constrained optimal control problems. The resulting optimization problems are non-convex and can be solved by direct methods \cite{lee2013quaternion}, indirect methods \cite{lee2016geometric}, global optimization techniques \cite{eren2015mixed,richards2002spacecraft,tam2016constrained}, or by applying convexification techniques \cite{kim2004quadratically,hutao2011rhc}. These approaches are able to achieve high performance while accounting for the system dynamics, but have the disadvantage of being computationally intensive.

Potential field methods generate a suitable control law by constructing an artificial potential field that has a minimum at the target reference and high values around the exclusion zones. This field is then used as either a Lyapunov function to which techniques from nonlinear control are applied\cite{mcinnes1994large,lee2014feedback,avanzini2009potential,wisniewski2005slew,mengali2004spacecraft,shen2017velocity} or as a cost function to which techniques from optimal control\cite{spindler2002attitude} are applied. Potential field methods are typically computationally inexpensive, and thus suitable for onboard implementation. Under suitable assumptions they may also automatically provide stability and robustness guarantees. However, special care must be taken to prevent the spacecraft from becoming stuck in local minima or critical points of the potential function \cite{APF2,okamoto2015novel,doria2013algorithm}. A variety of potential field methods have been proposed in the literature, e.g. \cite{mcinnes1994large,lee2014feedback,avanzini2009potential} which use Gaussian functions, log barriers, and exponential functions, respectively, to encode forbidden and mandatory regions. However, while these papers consider the effect of stationary points induced by the artificial potential functions they do not account for actuator saturation. 

Finally, anti-windup schemes can be used to augment a nominal control to deal with control saturation \cite{Boada2010} and/or output constraints \cite{burlionIJC2017}. Although these methods are attractive to practitioners since they do not modify the nominal control law, addressing multiple exclusion cone constraints with this approach can be very challenging.

In this paper we propose a new method based on the Explicit Reference Governor (ERG) framework \cite{garone2016explicit,ERG_CSM}. An ERG is an add-on unit which manipulates the commands given to a pre-stabilized system to prevent constraint violation. The ERG consists of a navigation component, which generates a kinematically feasible path towards the target, and a dynamic safety margin mechanism, which prevents constraint violations in transients when following this path with the closed-loop dynamical system. 

Compared to existing methods, the ERG developed in this paper (i) enforces control saturation and multiple exclusion cone\footnote{Although this paper focuses on the case of exclusion zone constraints, we note that it is straightforward to extend the proposed methodology to the case of inclusion cone constraints.} constraints, (ii) rigorously handles stationary points of the artificial potential function using a saddle destabilization term, and (iii) is computationally inexpensive due to its closed-form formulation. The theoretical contributions of this paper are twofold. We apply the theory of ERGs to quaternion spaces by constructing an appropriate navigation field and dynamic safety margin then proving their admissibility. In addition, we introduce a novel procedure, based on Nagumo's theorem, to construct less conservative dynamic safety margins for handling actuator saturation constraints.



\section{Preliminaries}
This paper will employ elements belonging to the Cartesian space $\real^n$, the unit vector set $\unit^n:=\{u\in\real^n~|~u^T\!u=1\}$, and the unitary quaternion set $\quat$. Given $q\in\quat$, we denote $q_R=\text{Re}(q)$ and $q_I=\text{Im}(q)$ as the real and imaginary components of $q$, respectively. Moreover, we denote by $q^*$ the complex-conjugate of $q$. With a slight abuse of notation, quaternions will sometimes be treated as elements of $\unit^4$ for the sake of using matrix multiplication, thus leading to $q_R\in\real$ for the scalar part of the quaternion and $q_I\in\real^3$ for its vector part. The complex conjugate of a quaternion $q$ satisfies $q_R^*=q_R$ and $q_I^*=-q_I$. Given two quaternions $q,p\in\quat$, the quaternion product $s=qp$ can be computed in matrix form using
\begin{equation}\label{eq:QuatProduct}
\begin{bmatrix} s_R \\ s_I \end{bmatrix} = \begin{bmatrix}
q_R & -q_I^T \\
q_I & q_R I_3+\hat{q}_I
\end{bmatrix} \begin{bmatrix} p_R \\ p_I \end{bmatrix},
\end{equation}
where the operator $\wedge:\real^3\to\real^{3\times3}$ is defined as
\begin{equation}
\hat{w} = \begin{bmatrix}
~~0 & -w_3 & ~~w_2 \\
~~w_3 & ~~0 & -w_1 \\
-w_2 & ~~w_1 & ~~0
\end{bmatrix}\!.
\end{equation}
Given a quaternion $q\in\quat$, the rotation matrix that transforms vector representations from the body-fixed reference frame to the inertial reference frame is given by
\begin{equation}\label{eq:Rodriguez}
R(q)=(q_R^2-q_I^T\!q_I)I_3+2 q_Iq_I^T+2q_R\hat{q}_I.
\end{equation}
For further details on quaternion algebra and their use in representing SO(3), the reader is referred to, e.g. \cite{SO3_Representation}. 

\section{Modeling and Problem Statement}
Consider a rigid body, e.g. a spacecraft, and let $q\in\quat$ represent the orientation of the body reference frame with respect to an inertial reference frame. Let the angular velocities $\omega\in\real^3$ and the control torques $\tau\in\real^3$ be expressed in the body reference frame. Following from the Newton-Euler equations, the rotational dynamics of the rigid body is given by the state-space model
\begin{subequations}\label{eq:Satellite}
  \begin{align}[left ={\empheqlbrace}~]
    & 2\dot{q}=E(q)\omega\label{eq:SatKin}\\
    & J\dot\omega=-\hat\omega J\omega+\tau\label{eq:SatDyn},
   \end{align}
\end{subequations}
where $J>0$ is the inertia matrix and
\begin{equation}\label{eq:QuatDerivative}
E(q) = \begin{bmatrix}
-q_I^T \\
q_R I_3-\hat{q}_I
\end{bmatrix}
\end{equation}
is the quaternion kinematic differentiation matrix. The system \eqref{eq:Satellite} is subject to the following non-convex set of input and state constraints.

\textbf{Actuator Saturation}: The actuator saturation constraint accounts for the fact that the control input is subject to physical limitations specific to thruster-based or momentum exchange attitude control systems. These constraints can be expressed using the following element-wise inequalities,
\begin{equation}\label{eq:ActSat_Original}
-\tau_{\max}\leq\tau\leq\tau_{\max},
\end{equation}
where $\tau_{\max}\in\real^3$ is a vector of positive values.

\textbf{Exclusion Cones}: The exclusion cone constraints are introduced to prevent certain on-board instruments, e.g. high sensitivity optical sensors, from pointing to an undesired direction, e.g. the sun. To this end, let $e_i\in\unit^3$, $i\in\{1,\ldots,l\}$, be a collection of unitary vectors defined in the body reference frame and let $h\in\unit^3$ be an undesired heading defined in the inertial reference frame. Given the current spacecraft orientation $q\in\quat$, the angular distance between the undesired heading $h$ and the sensor orientation $R(q)\,e_i$ can be lower bounded by  the constraint
\begin{equation}\label{eq:ExclusionCone}
h^T\!R(q)\,e_i\leq\cos\psi_i,\quad i\in\{1,\ldots,l\},
\end{equation}
where $\psi_i$ is the minimum admissible angle, e.g. the half conic aperture of the $i$-th sensor. To prevent overlap between any two exclusion cone constraints, the following assumption is made:

\begin{ass}\label{ass:DistanceMargin}
There exists an influence margin $\zeta\in(0,\pi)$ such that the orientations of all the onboard sensors satisfy
\begin{equation}
e_i^Te_j<\cos(\psi_i+\psi_j+2\zeta),\quad\forall i\neq j,
\end{equation}
with $i,j\in\{1,\ldots,l\}$.\sqr
\end{ass}

Given the influence margin $\zeta\in(0,\pi)$, the set of unaffected orientations will be denoted by 
\begin{equation}\label{eq:Unaffected}
\mathcal{R}_\zeta:=\left\{v\!\in\!\quat\left|~h^T\!R(v)e_i\!\leq\!\cos(\psi_i\!+\!\zeta),\,\forall i\!\in\!\{1,\ldots,l\}\!\right.\!\right\}\!.
\end{equation}
Analogously, given a static safety margin $\delta\in(0,\zeta)$, the set of steady-state admissible orientations will be denoted by
\begin{equation}\label{eq:SSAdmissible}
\mathcal{R}_\delta:=\left\{v\!\in\!\quat\left|~h^T\!R(v)e_i\!\leq\!\cos(\psi_i\!+\!\delta),\,\forall i\!\in\!\{1,\ldots,l\}\!\right.\!\right\}\!.
\end{equation}
Note that, by construction, the sets \eqref{eq:Unaffected}-\eqref{eq:SSAdmissible} satisfy $\mathcal{R}_\zeta\subset\mathcal{R}_\delta\subset\quat$.

\subsection{Problem Statement}
The objective of this paper is to develop a control strategy that solves the following constrained control problem:

\begin{prb}\label{prb:General}
Let system \eqref{eq:Satellite} be subject to constraints \eqref{eq:ActSat_Original}-\eqref{eq:ExclusionCone}. Then, given suitable initial angular velocities $\omega(0)\in\real^3$ and a constant reference $r\in\quat$ strictly satisfying \eqref{eq:ExclusionCone} with $q = r$, design a control law such that:
\begin{enumerate}[1.]
\item Constraints \eqref{eq:ActSat_Original}-\eqref{eq:ExclusionCone} are always satisfied;
\item The equilibrium point $(q,\omega)=(r,0)$ is asymptotically stable and is attractive for any $q(0)\in\quat$ strictly satisfying \eqref{eq:ExclusionCone}.\sqr
\end{enumerate}
\end{prb}
 
To meet this objective, we propose a two-step solution based on the explicit reference governor framework. The first step is pre-stabilizing the system dynamics so that its orientation asymptotically tends to an auxiliary reference $v\in\quat$. This is done using the control layer detailed in Section \ref{sec: Control}, which does not account for the system constraints. The second step consists in manipulating the dynamics of the auxiliary reference $v(t)$ so that it asymptotically tends to $r$, while simultaneously ensuring constraint satisfaction. This is accomplished using the navigation layer detailed in Section \ref{sec: Navigation}.


\section{Control Layer}\label{sec: Control}
The objective of the control layer is to pre-stabilize the system dynamics. Consequently, this section assumes, for the purpose of controller design, that the auxiliary reference $v$ is constant. The time-varying nature of $v(t)$ will be addressed by the ERG framework detailed in Section \ref{sec: Navigation}. The following proposition summarizes some well-known results on quaternion-based unconstrained control. For further details, the reader is referred to, e.g., \cite{SO3_Control}.  

\begin{prp}\label{prp:Control}
Given system \eqref{eq:Satellite}, let $v\in\quat$ be a constant reference and let $\tilde{q}=qv^*$ be the attitude error. Then, given the control law
\begin{equation}\label{eq: Control Law}
\tau(\tilde{q},\omega) = -k_P\tilde{q}_I-k_D\omega,
\end{equation}
with $k_p,k_D > 0$, the equilibrium point $(q,\omega)=(v,0)$ of the closed-loop system is asymptotically stable and admits the set
\begin{equation}\label{eq:NoUnwinding}
\Omega:=\left\{(\tilde{q},\omega)~\left|~ 2k_P(1-\tilde{q}_R)+\frac12\omega^TJ\omega\leq2k_P\right.\right\}
\end{equation}
as an inner approximation of its attraction basin.
\sqr
\end{prp}
\begin{proof}
The proof follows \cite{khalil} by employing the closed-loop Lyapunov function
\begin{equation}\label{eq:Lyapunov}
V(\tilde{q},\omega) = 2k_P(1-\tilde{q}_R)+\frac12\omega^T\!J\omega,
\end{equation}
which converges to zero along the closed-loop trajectories for any initial condition satisfying $(\tilde{q}(0),\omega(0))\in\Omega$, thereby implying the result.
\end{proof} 

Given the control law \eqref{eq: Control Law}, the actuator saturation constraints \eqref{eq:ActSat_Original} can be rewritten as the following state-space constraints,
\begin{equation}\label{eq:ActSat}
\begin{array}{r}
k_P\tilde{q}_I+k_D\omega\leq\tau_{\max}\textcolor{white}{.},\\
-k_P\tilde{q}_I-k_D\omega\leq\tau_{\max}.
\end{array}
\end{equation}

\begin{rmk}
Limiting the initial conditions to $(\tilde{q}(0),\omega(0))\in\Omega$ is sufficient to avoid \emph{unwinding} effects \cite{SO3_Control}. Although this assumption may seem restrictive from the classic viewpoint of ``global'' stabilization, it is worth noting that the presence of constraints changes the nature of the problem to the extent that a global region of attraction is no longer attainable. 
\end{rmk}

\section{Navigation Layer}\label{sec: Navigation}
The objective of the navigation layer is to augment the control layer by providing constraint handling capabilities to the closed-loop system. This will be done using the explicit reference governor, which manipulates the auxiliary reference $v(t)$ so that the transient dynamics cannot cause constraint violation. Based on the intuition presented in \cite{ERG_CSM}, this paper will generate the auxiliary reference using the dynamic system
\begin{equation}\label{eq:ERG}
\dot{v}=\Delta(\tilde{q},\omega)\rho(v,r),
\end{equation}
where $\Delta:\quat\times\real^3\to\real$ is referred to as a \emph{Dynamic Safety Margin} and $\rho:\quat\times\quat\to\real^4$ is referred to as a \emph{Navigation Field}. A rigorous definition of these terms, as well as a formal proof of how the ERG ensures constraint satisfaction, can be found in \cite{ERG_CSM}. Note that, although \cite{ERG_CSM} is formulated in Cartesian space, the results stated therein can be implemented on quaternions with straightforward modifications. The chief difficulties lie in defining appropriate an dynamic safety margin and navigation field in $\quat$.

The dynamic safety margin can be interpreted as the distance between the boundary of the constraints \eqref{eq:ExclusionCone}, \eqref{eq:ActSat} and the trajectory that the closed-loop system \eqref{eq:Satellite}, \eqref{eq: Control Law} would follow if the current reference $v$ were to remain constant. This value is thus used to quantify how safe it is to modify the current reference without causing a constraint violation at any time in the future. The navigation field can instead be interpreted as the direction in which the auxiliary reference $v(t)$ should evolve to reach the desired reference $r$ while simultaneously remaining steady-state admissible. The following subsections will illustrate how to generate suitable $\Delta(\tilde{q},\omega)$ and $\rho(v,r)$.

\subsection{Dynamic Safety Margin}\label{ssec:DSM}
The objective of this subsection is to obtain a Lipschitz continuous scalar function $\Delta$ that satisfies the following properties:
\begin{enumerate}[1.]
\item \emph{Recursive Feasibility:} Whenever $\Delta(\tilde{q},\omega)\geq0$, where $\tilde{q}$ and $\omega$ are the attitude error and the angular velocity at the current time, it is possible to guarantee constraint satisfaction at any time in the future by not changing the current value of the auxiliary reference;
\item \emph{Forward Invariance:} Whenever $\Delta(\tilde{q},\omega)=0$, the closed-loop system satisfies $\dot\Delta(\tilde{q},\omega)\geq0$ as long as the auxiliary reference remains constant;
\item \emph{Strong Returnability:} For any constant and strictly steady-state admissible reference, $\Delta(\tilde{q},\omega)$ asymptotically tends to a value that is strictly positive.
\end{enumerate}
The first step in constructing the dynamic safety margin is to separate the contributions of the various constraints by defining
\begin{equation}\label{eq:DSM}
\Delta(\tilde{q},\omega)=\min\{\Delta_e(\tilde{q},\omega),\Delta_a(\tilde{q},\omega)\},
\end{equation}
where $\Delta_e(\tilde{q},\omega)$ is the dynamic safety margin associated to the exclusion cone constraints \eqref{eq:ExclusionCone}, whereas $\Delta_a(\tilde{q},\omega)$ is the dynamic safety margin associated to the actuator saturation constraint \eqref{eq:ActSat}.\medskip

\noindent \textbf{Exclusion Cones}: For what concerns the exclusion cone constraints, it is sufficient to note that, given $v\in\mathcal{R}_\delta$, constraints \eqref{eq:ExclusionCone} hold true for any attitude $q\in\quat$ such that $2\arccos\tilde{q}_R\leq\theta(v)$, where \begin{equation}\label{eq:MinDist}
\theta(v)=\min_{i\in\{1,\ldots,l\}}\left\{\arccos\left(h^T\!R(v)\,e_i\right)-\psi_i\right\}
\end{equation}
is the minimum angular distance between $v$ and the boundary of the exclusion cone constraints. Following the general approach in \cite{ERG_EL}, constraint satisfaction can thus be guaranteed by assigning the dynamic safety margin
\begin{equation}\label{eq:DSMexclusion}
\Delta_e(\tilde{q},\omega)=\kappa_e\left(\Gamma_e(v)-V(\tilde{q},\omega)\right),
\end{equation}
where 
\begin{equation}\label{eq:SatGamma}
\Gamma_e(v)=2k_P\left(1-\cos\tfrac{\theta(v)}{2}\right)
\end{equation}
can be interpreted as the potential energy associated to $\tilde{q}_R=\cos(\theta(v)/2)$. The intuition behind this choice is that, as long as \eqref{eq:DSMexclusion} remains non-negative, the closed-loop system will never have sufficient energy to violate the exclusion cone constraints. 

\begin{rmk}\label{rem:Unwinding}
An important consequence of equation \eqref{eq:MinDist} is that $\theta(v)$ is guaranteed to satisfy the upper bound $\theta(v)\leq\pi-\min\{\psi_i\}$. Since this entails $\Gamma_e(v)<2k_P$, the proposed constrained control strategy has the added effect of ensuring $(\tilde{q},\omega)\in\Omega$ at all times, thus preventing unwinding phenomena.
\end{rmk}\medskip

\begin{figure}
\centering
\includegraphics[width=0.75\columnwidth]{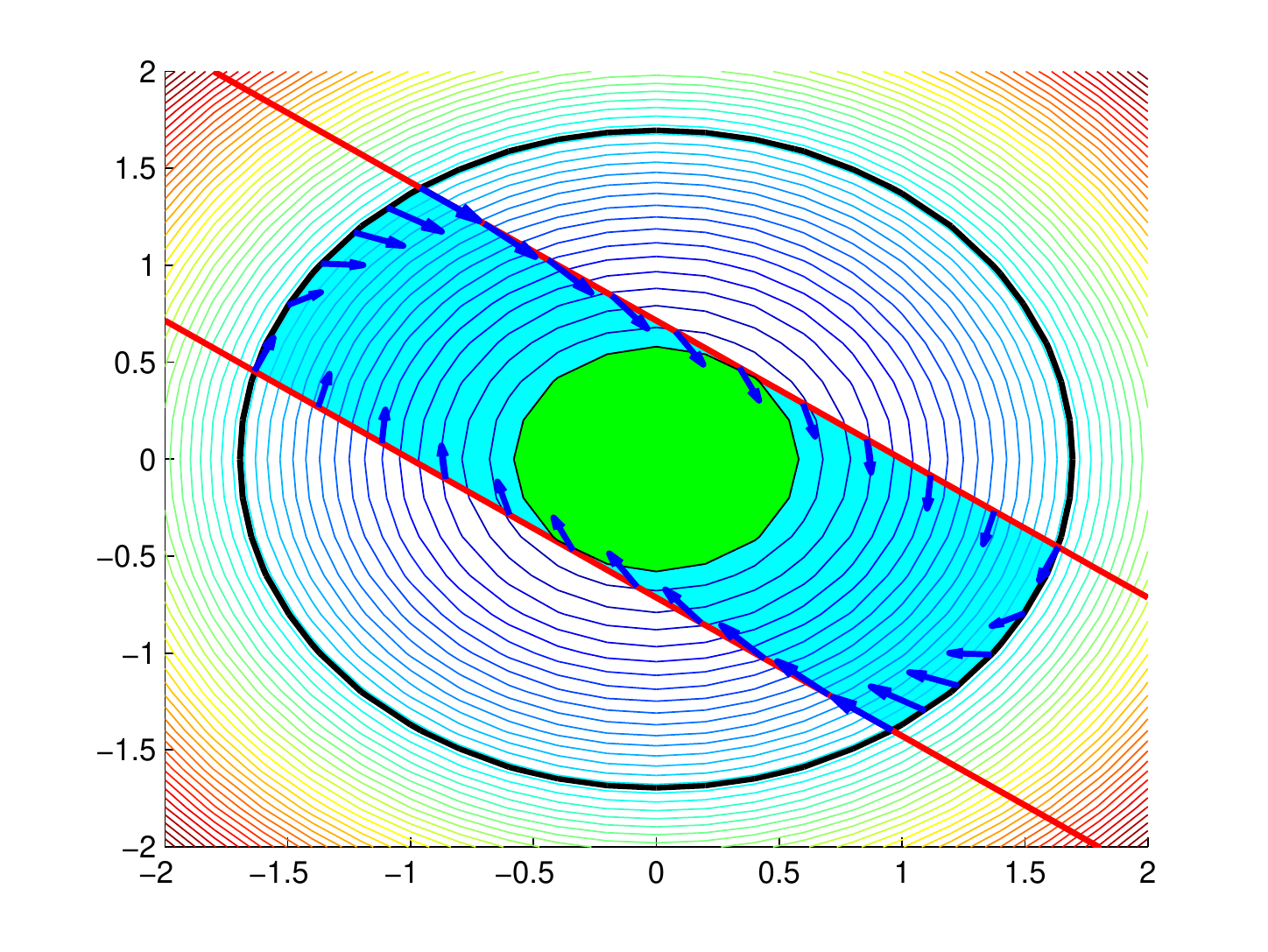}
\caption{Qualitative representation of the invariant set \eqref{eq:ActInvSet}, shown in cyan. The set boundaries are $\|\tau\|=\tau_{\max}$, in red, and $V=\Gamma_a$, in black. The invariance of the set can be deduced from the fact that the system trajectories point inward, thus satisfying the Nagumo theorem \cite{blanchini}. For the sake of comparison, the invariant set obtained by solving \eqref{eq:ActGamma_No} is reported in green; it is smaller and hence a dynamic safety margin based on it will lead to more conservative performance.}\label{fig:invariant}
\vspace{-0.5cm}
\end{figure}

\noindent \textbf{Actuator Saturation}: Following the general approach outlined in \cite{ERG_CSM}, the actuator saturation constraint \eqref{eq:ActSat} can be enforced by limiting the value of the Lyapunov function \eqref{eq:Lyapunov} so that the corresponding level-set does not exit the plane $k_P\tilde{q}_I+k_D\omega+\tau_{max}=0$. This can be done by solving the optimization problem 
\begin{equation}\label{eq:ActGamma_No}
\begin{cases}
\min & 2k_P(1-\tilde{q}_R)+\omega^TJ\omega\\
\text{s.t.} & \tilde{q}_R^2+\tilde{q}_I^T\tilde{q}_I=1\\
&\bigl[-k_P\tilde{q}_I-k_D\omega\bigr]_i\geq\tau_{\max,i}
\end{cases}
\end{equation}
and taking the minimum for $i=\{1,2,3\}$. Depending on the inertia matrix $J$ and the control gains $k_P$, $k_D$, however, this approach may lead to an excessively conservative response. In this paper we introduce a novel method for constructing a larger level-set to improve performance. We propose the dynamic safety margin
\begin{equation}\label{eq:ActDelta}
\Delta_a(\tilde{q},\omega)=\min\{\kappa_\tau(\tau_{\max}\!-\!\tau(\tilde{q},\omega)),\kappa_a(\Gamma_a\!-\!V(\tilde{q},\omega))\},
\end{equation}
where $\tau$ and $V$ are given in equations \eqref{eq: Control Law} and \eqref{eq:Lyapunov}, respectively, whereas $\Gamma_a$ is a positive scalar that can be computed by solving
\begin{equation}\label{eq:ActGamma}
\begin{cases}
\min & 2k_P(1-\tilde{q}_R)+\omega^TJ\omega\\
\text{s.t.} & \tilde{q}_R^2+\tilde{q}_I^T\tilde{q}_I=1\\
&\bigl[-k_P\tilde{q}_I-k_D\omega\bigr]_i=\tau_{\max,i}\\
&\bigl[-k_P\tilde{q}_I-k_D\omega\bigr]_j\leq\tau_{\max,j},\qquad j=\{1,2,3\}\setminus i\\
&\begin{bmatrix}
\frac12E(\tilde{q})\omega&
J^{-1}(-\hat{\omega}J\omega+\tau)
\end{bmatrix}_i\begin{bmatrix}
k_P\\
k_D
\end{bmatrix}\leq0,
\end{cases}
\end{equation}
where $[\cdot]_i$ denotes the ith row of a matrix, for $i=\{1,2,3\}$, and taking the smallest of the three solutions. As detailed\footnote{Due to the complexity of representing $\tilde{q}$ and $\omega$ in a 2D plot, Figure \ref{fig:invariant} was obtained using analogous considerations for the second-order system $\ddot{x}=\tau$ with $\tau=-k_Px-k_D\dot{x}$. Note that, although the numerical values are not indicative for SO(3), the overall behavior is very similar.} in Figure \ref{fig:invariant}, the idea behind equation \eqref{eq:ActGamma} is to define $\Gamma_a$ as the largest value of the Lyapunov function such that
\begin{equation}
\tau_i(\tilde{q},\omega)=\tau_{\max,i}~\wedge~V(\tilde{q},\omega)\leq\Gamma_a~~\Rightarrow~~ \nabla\tau_i\left[\begin{smallmatrix}
\dot{q}\\
\dot\omega
\end{smallmatrix}\right]\leq0,
\end{equation}
for $i=\{1,2,3\}$.  This property, combined with the time-decreasing nature of the Lyapunov function along the system trajectories, is sufficient to ensure that the set
\begin{equation}\label{eq:ActInvSet}
\{\tilde{q},\omega|~ \tau(\tilde{q},\omega)\!\in\![-\tau_{\max},\tau_{\max}]\}~\bigcap~\{\tilde{q},\omega|~V(\tilde{q},\omega)\!\leq\!\Gamma_a\}
\end{equation}
is forward invariant for any constant reference $v$. As a result, it follows from \cite{ERG_CSM} that \eqref{eq:ActDelta} is a dynamic safety margin since it measures the distance between the current state $(\tilde{q},\omega)$ and the boundary of the invariant set \eqref{eq:ActInvSet}.




\subsection{Navigation Field}
The objective of this section is to construct a vector field that, for any initial auxiliary reference $v(0)\in\mathcal{R}_\delta$, generates a steady-state admissible path that leads to the desired reference $r\in\mathcal{R}_\zeta$. To do so, the first step will be to account for the kinematic behavior of quaternions by defining
\begin{equation}\label{eq:NavigationField}
\rho(v,r)=E(v)\bigl(\rho_r(v,r)+\rho_e(v)+\rho_d(v,r)\bigr),
\end{equation}
where $E(v)\in\real^{4\times3}$ is defined in \eqref{eq:QuatDerivative}, $\rho_r:\mathcal{R}_\delta\times\mathcal{R}_\zeta\to\real^3$ is an attraction term that points from $v$ to $r$, $\rho_r:\mathcal{R}_\delta\to\real^3$ is a repulsion term that points away from the constraints, and $\rho_d:\mathcal{R}_\delta\times\mathcal{R}_\zeta\to\real^3$ is a destabilization term that prevents stagnation. Each of these terms will be addressed using a step-by-step analysis based on the artificial potential field approach \cite{APF}. 
\medskip

\noindent \textbf{Attraction Term}: To design the attraction term, consider the case in which the system is not subject to constraints.
In this scenario, a simple choice would be $\rho_r(v,r) = -\tilde{v}_I$, where $\tilde{v}=vr^*$ is the reference error. For the sake of modularity, however, in this paper we propose a different choice,
\begin{equation}\label{eq: Attraction Field}
\rho_r(v,r) = - \frac{\tilde{v}_I}{\max\{\|\tilde{v}_I\|,\sin\frac\eta2\}},
\end{equation}
where $\eta\in(0,\pi)$ is a (preferably small) smoothing angle. The advantage of \eqref{eq: Attraction Field} is that $\|\rho_r(v,r)\|=1$ for all $\tilde{\alpha}\in(\eta,\pi]$, where $\tilde\alpha:=2\arccos(\tilde{v}_R)$ is the angular error between $v$ and $r$. In the presence of exclusion cone constraints, the proposed attraction term can be combined with a repulsion term to ensure constraint satisfaction. Note that, under our assumptions, the input saturation constraints are always satisfied at steady-state, meaning that they do not require any additional terms.

\medskip
\noindent\textbf{Repulsion Term}: The objective of the repulsion term is to ensure that $v(t)$ does not exit the set of steady-state admissible values $\mathcal{R}_\delta$. By taking advantage of the fact that the attraction field satisfies $\|\rho_r(v,r)\|\leq1$, the proposed repulsion term has the form
\begin{equation}\label{eq: Repulsion Field, Exclusion Cone}
\rho_e(v) = \max\left\{\frac{\zeta-\theta(v)}{\zeta-\delta},0\right\}\frac{\hat{h}R(v)e_\mathcal{I}}{\|\hat{h}R(v)e_\mathcal{I}\|},
\end{equation}
where $\theta(v)$ is the minimum angular distance between $v$ and the edge of the closest exclusion cone, as detailed in \eqref{eq:MinDist}, and $\mathcal{I} \in \{1, \cdots, l\}$ is the associated index. Indeed, \eqref{eq: Repulsion Field, Exclusion Cone} satisfies $\|\rho_e(v)\|=1$ for $\theta(v)=\delta$, thus making it easier to study in relationship to $\rho_r(v,r)$. Equation \eqref{eq: Repulsion Field, Exclusion Cone} is also designed so that $\|\rho_e(v)\|=0$ for $\theta(v)\geq\zeta$. Due to Assumption \ref{ass:DistanceMargin}, this ensures that the effect of a given exclusion cone is always limited to its direct influence region. For the reader's convenience, the role of $\zeta,\eta,\delta$ is detailed in Figure \ref{fig:NavField}. 
%
%
The following proposition details the asymptotic behavior of the navigation field, $\rho(v,r)=E(v)(\rho_r(v,r)+\rho_e(v))$, for the special case in which there is only one exclusion cone constraint.

\begin{figure}
\centering
\includegraphics[width=0.55\columnwidth]{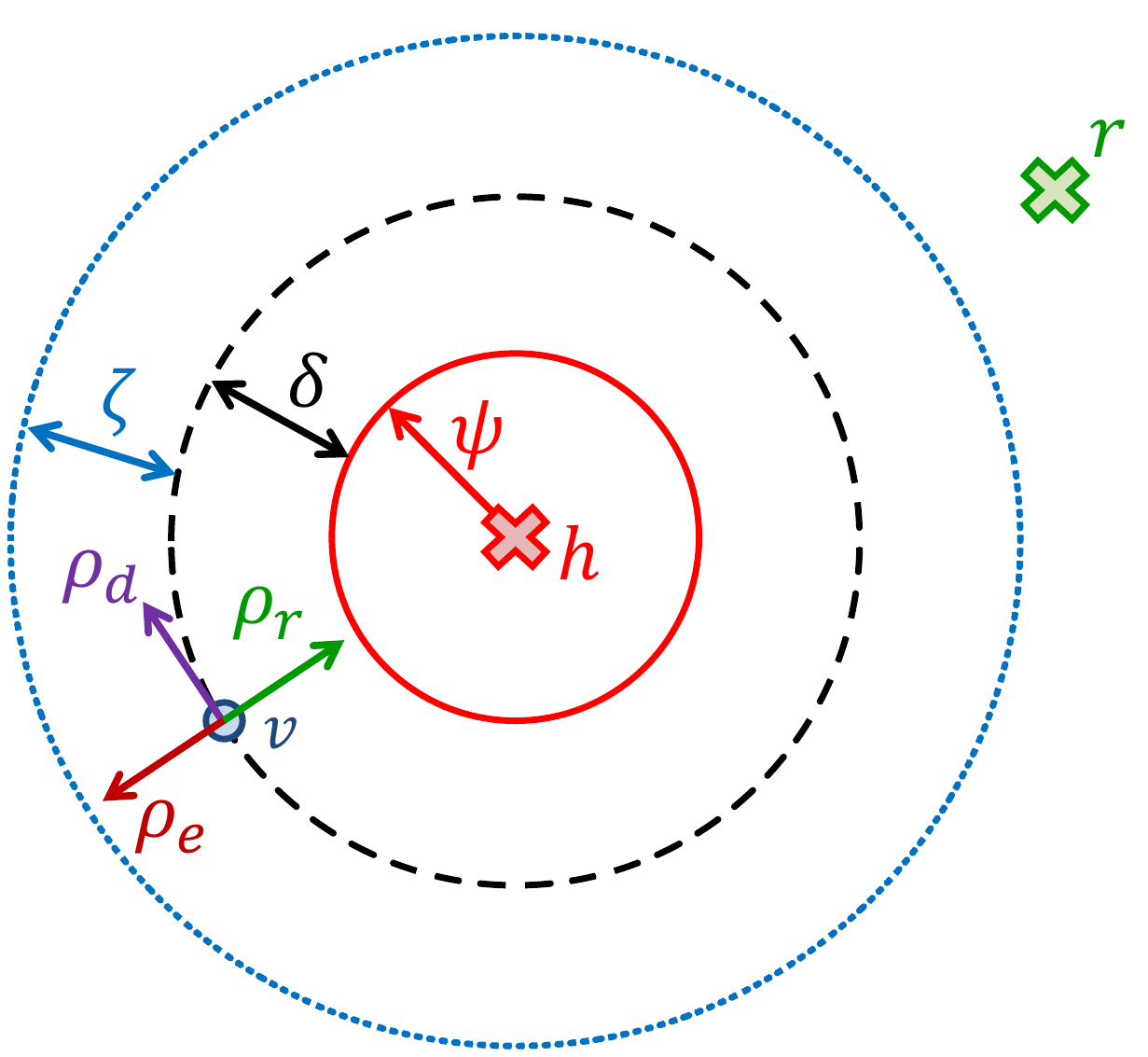}
\caption{Qualitative representation in $\real^2$ of the exclusion cone constraint and the various components in the navigation field \eqref{eq:NavigationField}. The red circle identifies the constraint boundary \eqref{eq:ExclusionCone}. The blue dotted line is the boundary of \eqref{eq:Unaffected}, which depends on the influence margin $\zeta$. The black dashed line is the boundary of \eqref{eq:SSAdmissible}, which depends on the static safety margin $\delta$. Within the influence region, $\rho_r$ is a term that points from $v$ to $r$, $\rho_e$ is a term that points away from the center of the exclusion cone constraint, and $\rho_d$ is a term that is always tangent to the constraint and does not increase the distance between $v$ and $r$. By construction, $\rho_d$ and $\rho_e$ are zero outside the influence margin $\zeta$ and have a modulus of one whenever the angular distance between $v$ and $r$ is equal to the static safety margin $\delta$.}\label{fig:NavField}
\end{figure}

\begin{prp}\label{prp: Exclusion}
Given a single exclusion cone constraint \eqref{eq:ExclusionCone}, i.e. $i=l=1$, let $\zeta\in(0,\pi)$ be an influence margin, let $\delta\in(0,\zeta)$ be a static safety margin, and let $\eta\in(0,\delta)$. Then, the system
\begin{equation}\label{eq:AttractionField_Conservative}
2\dot{v}=E(v)(\rho_r+\rho_e),
\end{equation}
with $\rho_r$ and $\rho_e$ given in \eqref{eq: Attraction Field} and \eqref{eq: Repulsion Field, Exclusion Cone} is such that, for any constant reference $r\in\mathcal{R}_\zeta$, the following properties hold:
\begin{enumerate}[1.]
\item The set $\mathcal{R}_\delta$ is forward-invariant;
\item The equilibrium point $s_1\in\quat$ satisfying
\begin{subequations}\label{eq:SaddleConditions}
\begin{eqnarray}
h^T\!R(s_1)e_1&=&\cos(\psi_1+\delta),\label{eq:SaddleDistance}\\
\frac{\hat{h}R(s_1)e_1}{\|\hat{h}R(s_1)e_1\|}\!\!&=&-\frac{\text{Im}(s_1r^*)}{\|\text{Im}(s_1r^*)\|},\label{eq:SaddleDirection}
\end{eqnarray}
\end{subequations}
is asymptotically stable with respect to the manifold 
\begin{equation}\label{eq:SaddleManifold}
\mathcal{S}_1:=\left\{v\!\in\!\quat~\left|~\begin{array}{rcl}
h^T\!R(v)e_1&\leq&\cos\psi_1\\
\displaystyle\frac{\hat{h}R(v)e_1}{\|\hat{h}R(v)e_1\|}\!\!&=&-\displaystyle\frac{\text{Im}(vr^*)}{\|\text{Im}(vr^*)\|}
\end{array}\right.\right\};
\end{equation}
\item The equilibrium point $v=r$ is asymptotically stable and admits $\mathcal{R}_\delta\setminus\mathcal{S}_1$ as a basin of attraction.
\end{enumerate}
\end{prp}
 
\begin{proof}
See Appendix \ref{app:Conservative}.
\end{proof}

The main limitation of Proposition \ref{prp: Exclusion} is that, even if the reference is strictly steady-state admissible and there is only one exclusion cone constraint, the trajectories of $v$ resulting from the proposed attraction field are not guaranteed to converge to $r$ for any admissible initial condition $v$. This is due to the presence of the saddle point $s_1$ which is characterized by an attraction basin of dimension zero. In the presence of multiple exclusion cones, this issue becomes even more problematic due to the well-known limitation of conservative vector fields, which always contain at least as many stagnation points as the number of holes in the domain \cite{APF2}. This issue will be addressed in the following paragraph. \medskip

\noindent\textbf{Saddle Destabilization Term}: The objective of the destabilization term is to eliminate the presence of saddle points without penalizing the convergence properties. To do so, the proposed vector field, $\rho_d(v,r)$, must satisfy the following properties:
\begin{itemize}
\item $\rho_d$ must be \textbf{non-conservative}, otherwise it will be unable to overcome the intrinsic limitations of conservative vector fields;
\item $\|\rho_d\|=0$ whenever $\theta(v)>\zeta$. This will ensure that the effects of the destabilization field are local;
\item $\rho_d^T\rho_r=0$, to avoid pointing away from the reference $r$;
\item $\rho_d^T\rho_e=0$, to avoid pointing towards the constraints \eqref{eq:ExclusionCone}.
\end{itemize}
Based on these criteria, the following is proposed:
\begin{equation}\label{eq: Destabilization Field}
\rho_d(v,r) \!=\! \max\!\left\{\!\frac{\zeta-\theta(v)}{\zeta-\delta}\frac{-\rho_r(v,r)^T\hat{h}R(v)e_\mathcal{I}}{\|\hat{h}R(v)e_\mathcal{I}\|},0\!\right\}\,\varphi,
\end{equation}
where the vector $\varphi\in\unit^3$ must satisfy
\begin{equation}\label{eq:DestabilizationDirection}
\varphi \in \textrm{Null}\left(\begin{bmatrix}
\rho_e^T(v)\\
\rho_r^T(v,r)
\end{bmatrix}\right)\!.
\end{equation}
Note that $\varphi$ is uniquely defined everywhere except the manifold \eqref{eq:SaddleManifold}, where $\rho_r(v,r)$ is parallel to $\rho_e(v)$. In this case, the element of the null space can be chosen randomly and leads to a discontinuous, and thus non-conservative, vector field $\rho_d(v,r)$. The following proposition states that \eqref{eq: Destabilization Field} is successful at destabilizing the saddle points, thus making the desired reference $r$ globally asymptotically stable with respect to the domain $\mathcal{R}_\delta$.

\begin{prp}\label{prp: Destabilization}
Given a single exclusion cone constraint \eqref{eq:ExclusionCone}, i.e. $i=l=1$, let $\zeta\in(0,\pi)$ be an influence margin, let $\delta\in(0,\zeta)$ be a static safety margin, and let $\eta\in(0,\delta)$. Then, the system
\begin{equation}\label{eq:AttractionField_Global}
2\dot{v}=E(v)(\rho_r(v,r)+\rho_e(v)+\rho_d(v,r)),
\end{equation}
with $\rho_r$, $\rho_e$, and $\rho_d$ given in \eqref{eq: Attraction Field}, \eqref{eq: Repulsion Field, Exclusion Cone}, and \eqref{eq: Destabilization Field} is such that:
\begin{enumerate}[1.]
\item The set $\mathcal{R}_\delta$ is forward-invariant;
\item For any constant reference $r\in\mathcal{R}_\zeta$, the equilibrium point $v=r$ is globally asymptotically stable with respect to the domain $\mathcal{R}_\delta$.
\end{enumerate}
\end{prp}
 
\begin{proof}
See Appendix \ref{app:Global}.
\end{proof}
 
Proposition \ref{prp: Destabilization} states that, in the presence of a single exclusion cone constraint, the proposed vector field results in a trajectory that asymptotically tends to $r\in\mathcal{R}_\zeta$ for any initial condition $v(0)\in\mathcal{R}_\delta$. The following corollary extends this result to the case of multiple exclusion cone constraints.

\begin{cor}\label{cor:NF}
Given the exclusion cone constraints \eqref{eq:ExclusionCone}, let the influence margin $\zeta\in(0,\pi)$ satisfy Assumption \ref{ass:DistanceMargin}, let $\delta\in(0,\zeta)$ be a static safety margin, and let $\eta\in(0,\delta)$. Then, the system \eqref{eq:AttractionField_Global} with $\rho_r$, $\rho_e$, and $\rho_d$ given in \eqref{eq: Attraction Field}, \eqref{eq: Repulsion Field, Exclusion Cone}, and \eqref{eq: Destabilization Field} is such that:
\begin{enumerate}[1.]
\item The set $\mathcal{R}_\delta$ is forward-invariant;
\item For any constant reference $r\in\mathcal{R}_\zeta$, the equilibrium point $v=r$ is GAS with respect to the domain $\mathcal{R}_\delta$.
\end{enumerate}
\end{cor}
 
\begin{proof}
See Appendix \ref{app:Corollary}.
\end{proof}

\section{Main Result}
The following theorem combines all the previous results to formulate a suitable constrained control strategy.

\begin{trm}
Given system \eqref{eq:Satellite} subject to the actuator saturation \eqref{eq:ActSat_Original} and the exclusion cone constraint \eqref{eq:ExclusionCone}, let \eqref{eq: Control Law} with $k_P,k_D>0$ be the primary control layer, and let \eqref{eq:ERG} be the navigation layer subject to the dynamic safety margin \eqref{eq:DSM}, with \eqref{eq:DSMexclusion} and \eqref{eq:ActDelta}, and the navigation field \eqref{eq:AttractionField_Global}, with \eqref{eq: Attraction Field}, \eqref{eq: Repulsion Field, Exclusion Cone}, and \eqref{eq: Destabilization Field}. Then, under Assumption \ref{ass:DistanceMargin} and for any initial condition such that there exists an auxiliary reference $v(0)\in\mathcal{R}_\delta$ satisfying $\Delta(\tilde{q}(0),\omega(0))\geq0$, the following hold:
\begin{enumerate}[1.]
\item For any reference $r(t)$, the system constraints are satisfied;
\item For any constant reference $r\in\mathcal{R}_\zeta$, the system trajectories will asymptotically converge to the equilibrium point $q=r$, $\omega=0$.
\end{enumerate}
\end{trm}
\begin{proof}
Following from Proposition \ref{prp:Control}, the closed-loop system \eqref{eq:Satellite}, \eqref{eq: Control Law} is asymptotically stable for any $v$ and any state that satisfies the constraint $(\tilde{q},\omega)\in\Omega$ defined in \eqref{eq:NoUnwinding}. By construction, $\Delta(\tilde{q},\omega)\geq0$ is sufficient to ensure $(\tilde{q},\omega)\in\Omega$ as well as the satisfaction of constraints \eqref{eq:ActSat_Original}-\eqref{eq:ExclusionCone}. Following from \eqref{eq:DSM}, \eqref{eq:DSMexclusion}-\eqref{eq:SatGamma}, and \eqref{eq:ActDelta}-\eqref{eq:ActGamma}, the set $\{(\tilde{q},\omega)|\Delta(\tilde{q},\omega)\geq0\}$ is forward-invariant for any constant $v$ and there exists $\epsilon>0$ such that $v\in\mathcal{R}_\delta$ implies $\Delta(v,0)>\epsilon$. As a result, \eqref{eq:DSM} satisfies the requirements of \cite[Definition 1]{ERG_CSM}. Moreover, it follows from Corollary \ref{cor:NF} that \eqref{eq:AttractionField_Global}, with \eqref{eq: Attraction Field}, \eqref{eq: Repulsion Field, Exclusion Cone}, and \eqref{eq: Destabilization Field}, satisfies the requirements of \cite[Definition 2]{ERG_CSM}. The remainder of the proof is therefore a direct result of \cite[Theorem 1]{ERG_CSM}.
\end{proof}

\section{Numerical Simulations}
The proposed control strategy is applied to a rigid spacecraft with the inertia matrix $J=\textrm{diag}([918~920~1365])~kgm^2$. The control layer is implemented with $k_P=918~Nm/rad$ and $k_D=3672~Nms/rad$, whereas the navigation layer is implemented with $\eta=5^o$, $\delta=5^o$, $\zeta=10^o$, and $\kappa=10$. The following subsections will address different scenarios that focus on different aspects of the constrained control problem.
\begin{figure}
\centering
\includegraphics[trim=1.5cm 3cm 2.5cm 2cm, clip=false, width=0.95\columnwidth]{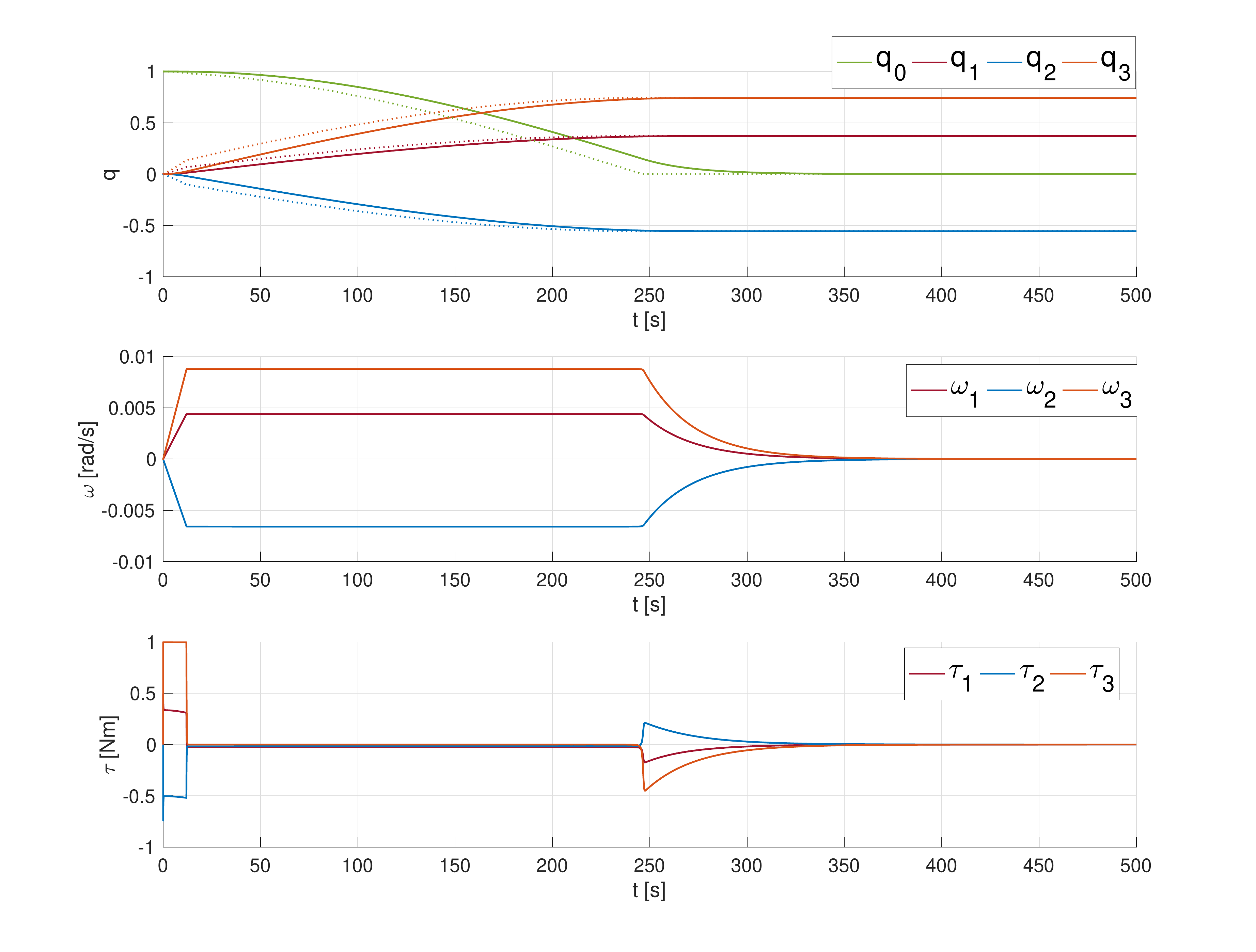}
\caption{Closed-loop response in the presence of only actuator saturation constraints. \textbf{Top:} The satellite orientation (solid) and auxiliary reference (dotted) converge to the desired steady-state. \textbf{Middle:} The angular velocities closely resemble the behavior of a trapezoidal trajectory planner. \textbf{Bottom:} The control inputs satisfy the constraint $|\tau|\leq1\,Nm$.}\label{fig:ex1_qwu}
\end{figure}

\subsection{Actuator Saturation Constraints}\label{ssec:ActSat}
In this subsection we consider a spacecraft subject to the maximum torque $\tau_{\max}=[1~1~1]^TNm$ and no exclusion cone constraints. Figure \ref{fig:ex1_qwu} illustrates the behavior obtained using the constrained control framework introduced in this paper. The method successfully achieves constraint satisfaction by imposing the maximum admissible torque at the very beginning and then moving at a constant angular rate before decelerating when it reaches close proximity of the desired setpoint $r=[0~0.74~0.37~-0.56]$. For comparison, it is worth noting that the response time obtained with the classic Lyapunov-based ERG that relies on the solution of \eqref{eq:ActGamma_No} is two orders of magnitude slower at reaching the same setpoint.
As discussed in Subsection \ref{ssec:DSM}, this is due to the fact that \eqref{eq:ActDelta} relies on a much larger invariant set.

\subsection{Single Exclusion Cone Constraint}\label{ssec:Single}
\begin{figure}
\centering
\subfloat{\includegraphics[width=0.5\columnwidth]{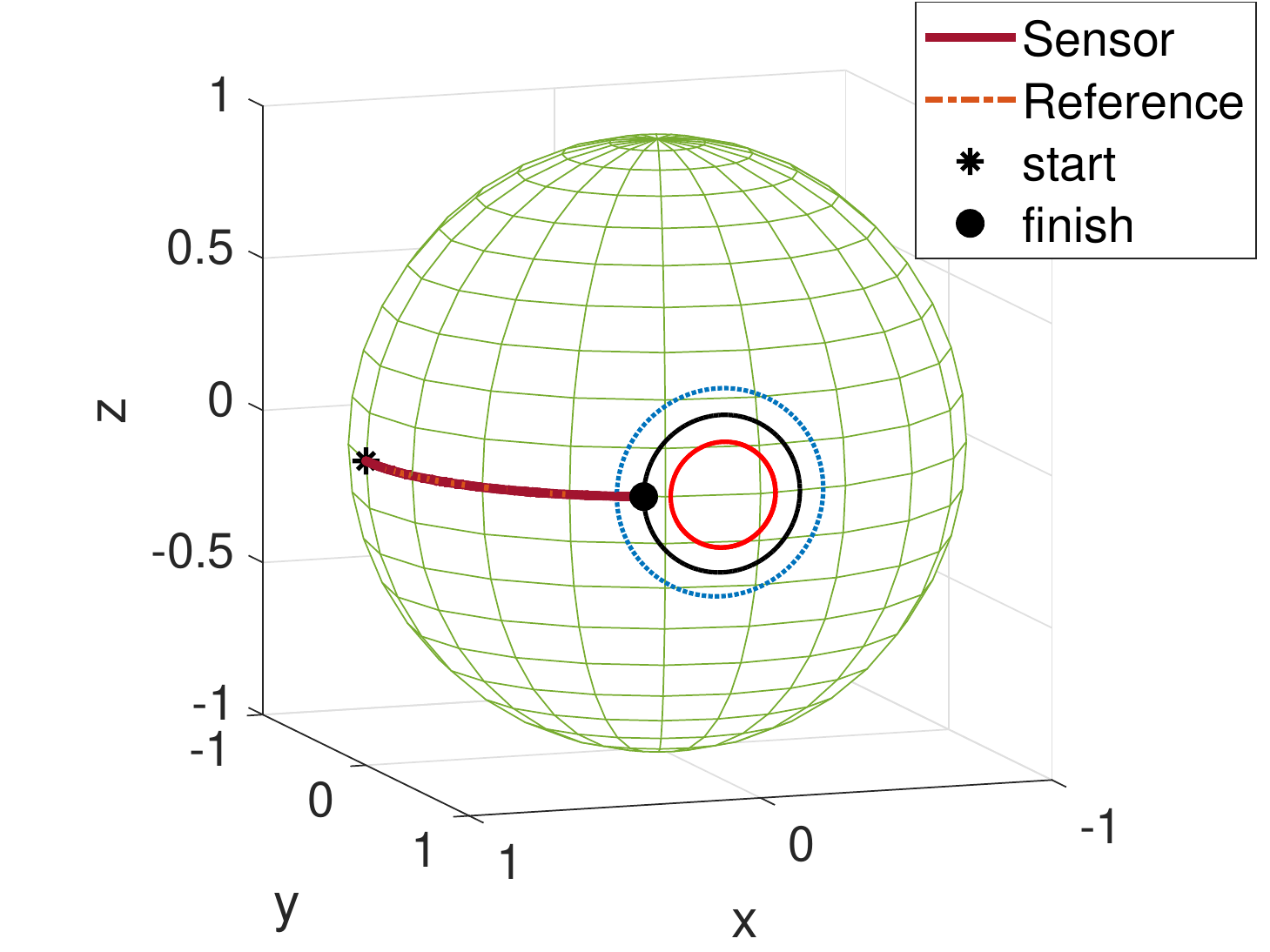}}
\subfloat{\includegraphics[width=0.5\columnwidth]{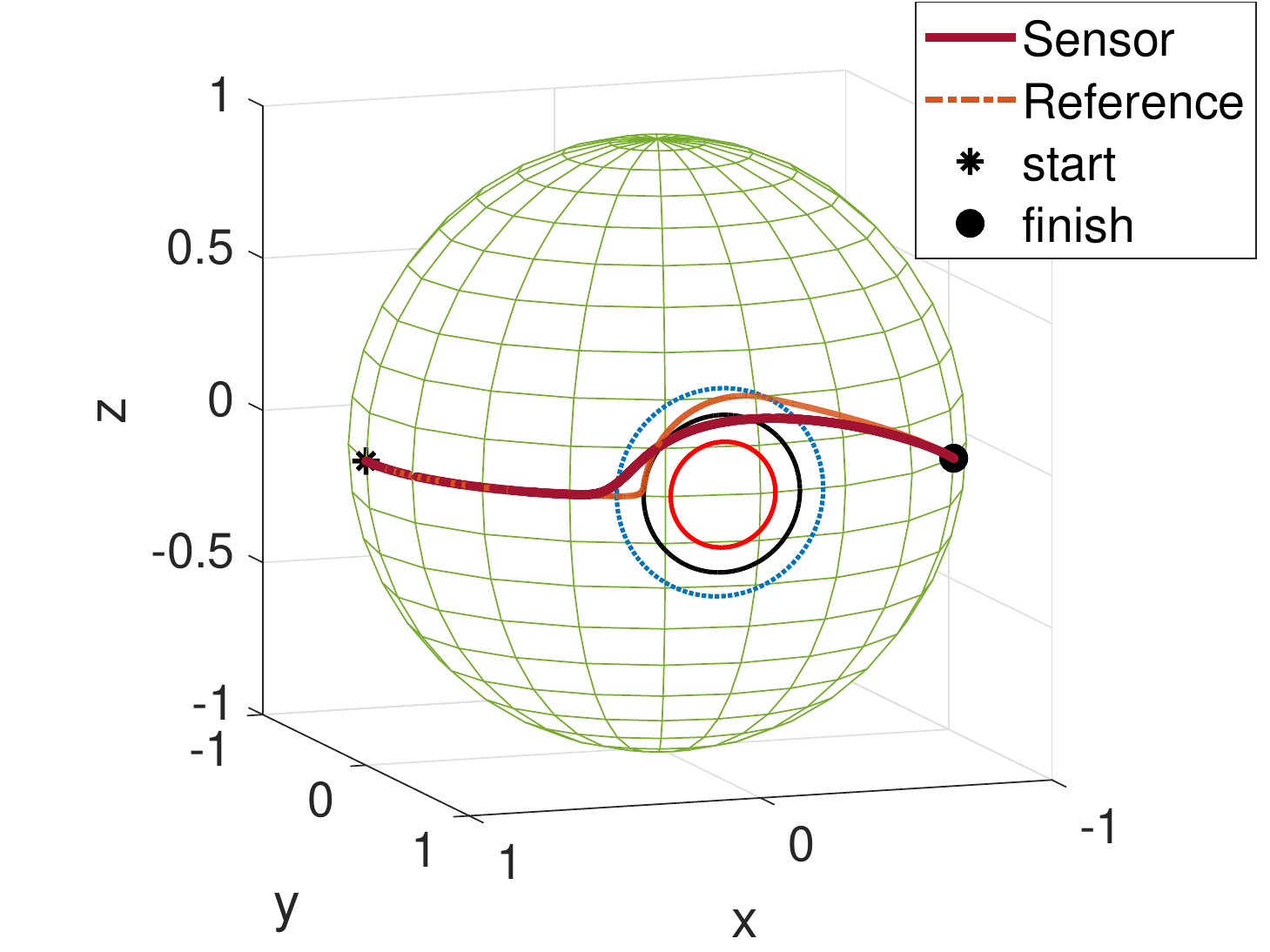}}
\caption{Closed-loop response of the satellite subject to an exclusion cone constraint. In the absence of the destabilization term (left), the system settles in an undesired equilibrium point. In the presence of the destabilization term (right) the obstacle is successfully overcome.\label{fig:ex2_sphr}}
\end{figure}
\begin{figure}
\centering
\includegraphics[trim=2cm 3cm 2.5cm 2cm, clip=true, width=0.95\columnwidth]{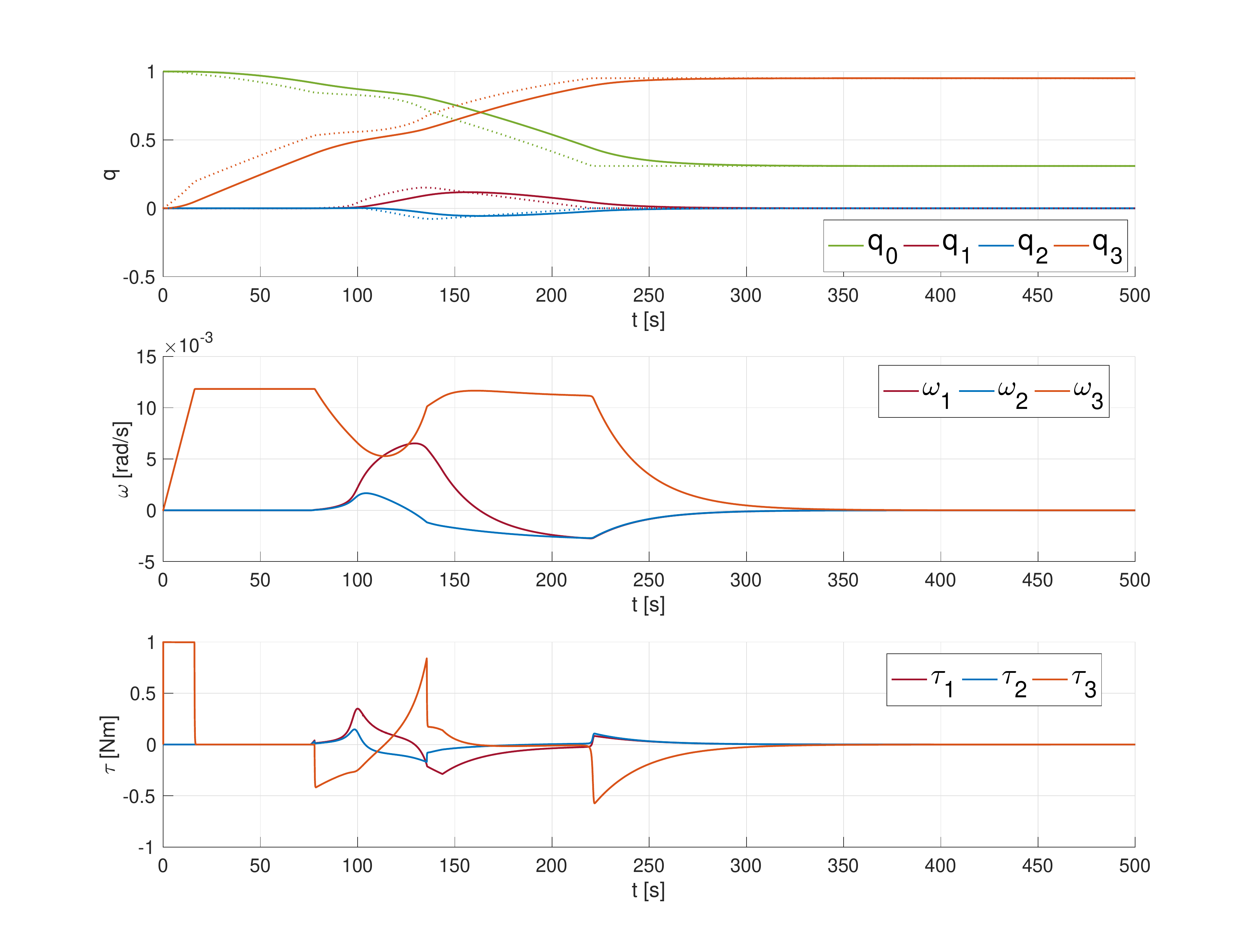}
\caption{Closed-loop response in the presence of actuator saturation constraints and a single exclusion cone constraint. The effect of the exclusion cone can be seen starting from time $t=78\,s$, which is when the auxiliary reference enters the influence region. \textbf{Top:} The satellite orientation (solid) and auxiliary reference (dotted) circumnavigate the exclusion cone and converge to the desired steady-state. \textbf{Middle:} Angular velocities. \textbf{Bottom:} The control inputs satisfy the constraint $|\tau|\leq1\,Nm$.}\label{fig:ex2_qwu}
\end{figure}
In this subsection we consider a satellite subject to the maximum torque $\tau_{\max}=[1~1~1]^TNm$ as well as an exclusion cone constraint. For the sake of constructing an example with known properties, the direction of the sun in the inertial reference frame is $h=[0~1~0]^T$ and the sensor orientation in the body reference frame is $e_1=[1~0~0]^T$. The half conic aperture of the sensor is $\psi_1=10^o$. Given $r=[0~0~0~1]^T$, it can be verified that $v(0)=[1~0~0~0]^T$ satisfies $v(0)\in\mathcal{S}_1$. Figure \ref{fig:ex2_sphr} illustrates the behavior obtained with and without the destabilization term \eqref{eq: Destabilization Field}. The satisfaction of the exclusion cone constraint can be seen by the fact that, in both cases, the sensor trajectories lie outside the red circle. The black circle represents the boundary of the set $\mathcal{R}_\delta$, which is shown to always contain the reference trajectories. Finally, the blue circle represents the border of the influence region. As expected from Proposition \ref{prp: Destabilization}, the destabilization term successfully prevents the ERG from converging to the undesired saddle point. Figure \ref{fig:ex2_qwu} illustrates the temporal behavior of the system state and inputs, which satisfy the imposed constraints.

\subsection{Multiple Exclusion Cones}\label{ssec:Multiple}
As a final simulation, we consider a spacecraft subject to the maximum torque $\tau_{\max}=[1~1~1]^TNm$ and two sensors, with $e_1=[0~0.9877~0.1564]^T$ and $e_2=[-0.04755~0.6545~0.5878]^T$, which must not be pointed in the direction $h=[0.7208~0.5237~0.4540]^T$. For ease of representation, the sensors have the same half conic aperture $\psi_1=\psi_2=10^o$. As illustrated in Figure \ref{fig:ex3_sphr}, the proposed method successfully generates a trajectory that orients the satellite to its desired reference while simultaneously ensuring constraint satisfaction.

\begin{figure}
\centering
\includegraphics[width=0.65\columnwidth]{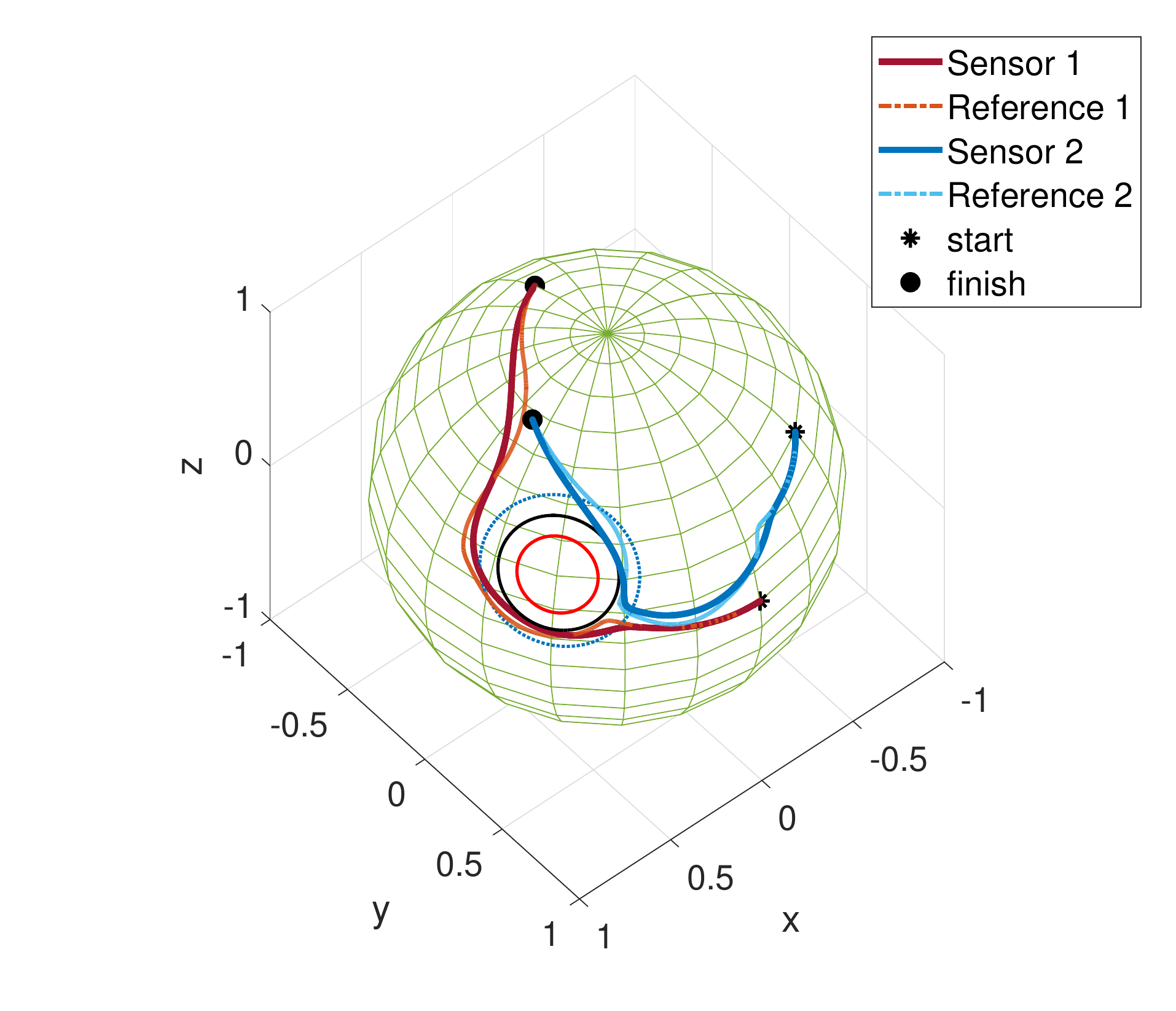}
\caption{Sensor trajectories in spherical and projected coordinates for the example addressed in Subsection \ref{ssec:Multiple}.}\label{fig:ex3_sphr}
\end{figure}

\begin{figure}
\centering
\includegraphics[trim=2cm 3cm 2.5cm 2cm, clip=true, width=0.95\columnwidth]{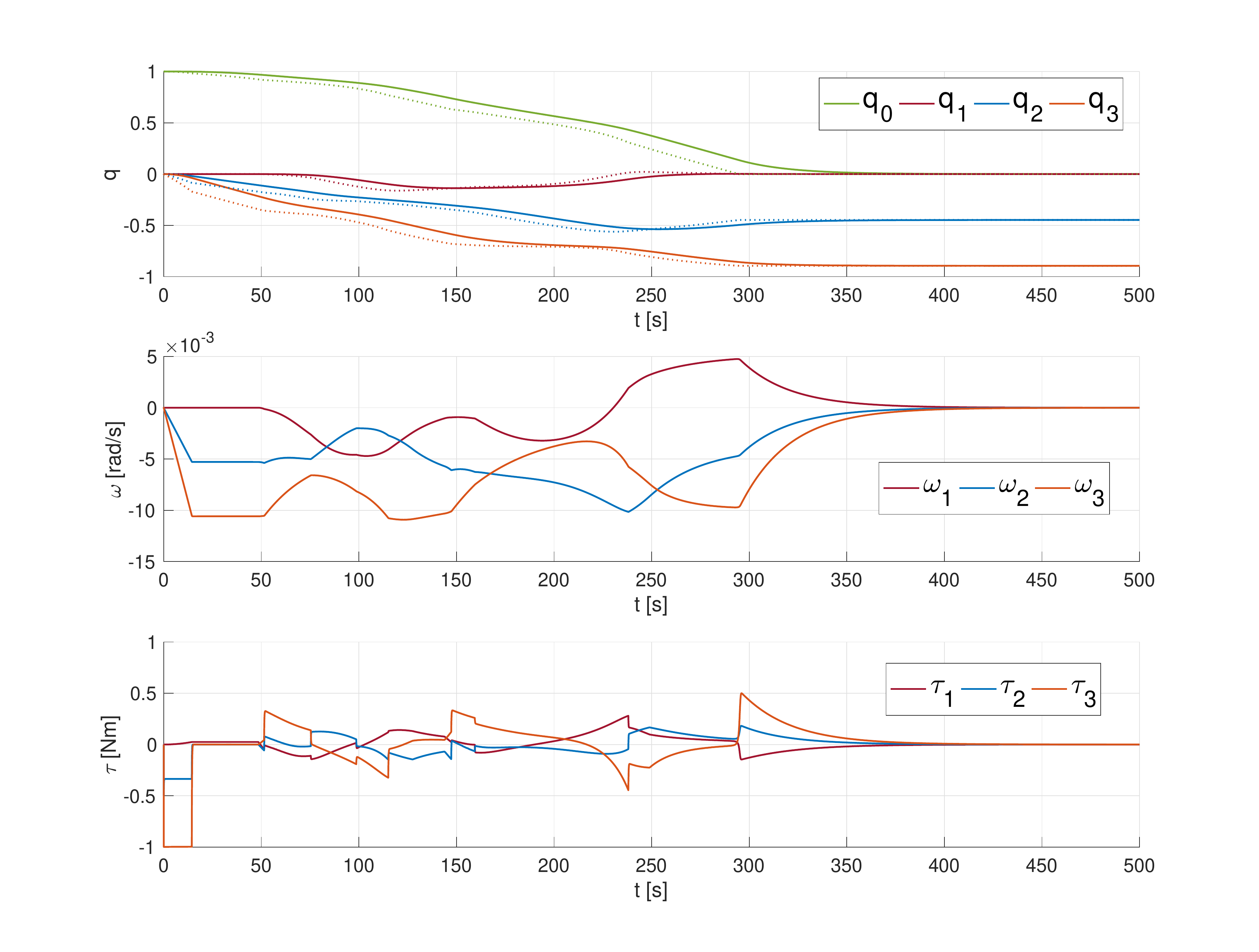}
\caption{Closed-loop response in the presence of actuator saturation and two exclusion cone constraints. \textbf{Top:} The satellite orientation (solid) and auxiliary reference (dotted) converge to the desired steady-state. \textbf{Middle:} Angular velocities. \textbf{Bottom:} The control inputs satisfy the constraint $|\tau|\leq1\,Nm$.}\label{fig:ex3_qwu}
\end{figure}

\section{Conclusions}
This paper has presented a novel approach to constrained spacecraft attitude control. The proposed method relies on the use of an explicit reference governor to manipulate the reference of a prestabilized system so that constraint satisfaction is guaranteed at all times. The proposed method successfully addresses actuator saturation and exclusion cone constraints in a simple and systematic manner that does not require online optimization. Numerical examples have been reported which illustrate the validity of the scheme. While guided to an extent by previous ERG theory introduced in Cartesian spaces, we note that both the construction of ERG for the quaternion-based attitude parameterization and the proofs of convergence properties are unique contributions of this paper. We have also introduced a novel procedure based on Nagumo's theorem to construct a less conservative dynamic safety margin. 

\bibliography{SO3}

\begin{thebibliography}{10}

\bibitem{koenig2009novel}
J.~D. Koenig, ``A novel attitude guidance algorithm for exclusion zone
  avoidance,'' in {\em Aerospace conference, 2009 IEEE}, pp.~1--10, IEEE, 2009.

\bibitem{hablani1999attitude}
H.~B. Hablani, ``Attitude commands avoiding bright objects and maintaining
  communication with ground station,'' {\em Journal of Guidance, Control, and
  Dynamics}, vol.~22, no.~6, pp.~759--767, 1999.

\bibitem{singh1997constraint}
G.~Singh, G.~Macala, E.~Wong, R.~Rasmussen, G.~Singh, G.~Macala, E.~Wong, and
  R.~Rasmussen, ``A constraint monitor algorithm for the cassini spacecraft,''
  in {\em Guidance, Navigation, and Control Conference}, p.~3526, 1997.

\bibitem{kjellberg2015discretized}
H.~C. Kjellberg and E.~G. Lightsey, ``Discretized quaternion constrained
  attitude pathfinding,'' {\em Journal of Guidance, Control, and Dynamics},
  vol.~38, no.~11, pp.~713--718, 2015.

\bibitem{tanygin2016fast}
S.~Tanygin, ``Fast autonomous three-axis constrained attitude pathfinding and
  visualization for boresight alignment,'' {\em Journal of Guidance, Control,
  and Dynamics}, vol.~40, no.~2, pp.~358--370, 2016.

\bibitem{weiss2014spacecraft}
A.~Weiss, F.~Leve, M.~Baldwin, J.~R. Forbes, and I.~Kolmanovsky, ``Spacecraft
  constrained attitude control using positively invariant constraint admissible
  sets on so (3)$\times$ R 3,'' in {\em American Control Conference (ACC),
  2014}, pp.~4955--4960, IEEE, 2014.

\bibitem{cui2007onboard}
P.~Cui, W.~Zhong, and H.~Cui, ``Onboard spacecraft slew-planning by heuristic
  state-space search and optimization,'' in {\em Mechatronics and Automation,
  2007. ICMA 2007. International Conference on}, pp.~2115--2119, IEEE, 2007.

\bibitem{frazzoli2001randomized}
E.~Frazzoli, M.~Dahleh, E.~Feron, and R.~Kornfeld, ``A randomized attitude slew
  planning algorithm for autonomous spacecraft,'' in {\em AIAA Guidance,
  Navigation, and Control Conference and Exhibit, Montreal, Canada}, 2001.

\bibitem{eren2017model}
U.~Eren, A.~Prach, B.~B. Ko{\c{c}}er, S.~V. Rakovi{\'c}, E.~Kayacan, and
  B.~A{\c{c}}{\i}kme{\c{s}}e, ``Model predictive control in aerospace systems:
  Current state and opportunities,'' {\em Journal of Guidance, Control, and
  Dynamics}, pp.~1--25, 2017.

\bibitem{lee2013quaternion}
U.~Lee and M.~Mesbahi, ``Quaternion-based optimal spacecraft reorientation
  under complex attitude constrained zones,'' in {\em AAS/AIAA Astrodynamics
  Specialist Conference. AAS/AIAA}, 2013.

\bibitem{lee2016geometric}
D.~Y. Lee, R.~Gupta, U.~V. Kalabi{\'c}, S.~Di~Cairano, A.~M. Bloch, J.~W.
  Cutler, and I.~V. Kolmanovsky, ``Geometric mechanics based nonlinear model
  predictive spacecraft attitude control with reaction wheels,'' {\em Journal
  of Guidance, Control, and Dynamics}, vol.~40, no.~2, pp.~309--319, 2016.

\bibitem{eren2015mixed}
U.~Eren, B.~A{\c{c}}{\i}kme{\c{s}}e, and D.~P. Scharf, ``A mixed integer convex
  programming approach to constrained attitude guidance,'' in {\em Control
  Conference (ECC), 2015 European}, pp.~1120--1126, IEEE, 2015.

\bibitem{richards2002spacecraft}
A.~Richards, T.~Schouwenaars, J.~P. How, and E.~Feron, ``Spacecraft trajectory
  planning with avoidance constraints using mixed-integer linear programming,''
  {\em Journal of Guidance, Control, and Dynamics}, vol.~25, no.~4,
  pp.~755--764, 2002.

\bibitem{tam2016constrained}
M.~Tam and E.~G. Lightsey, ``Constrained spacecraft reorientation using mixed
  integer convex programming,'' {\em Acta Astronautica}, vol.~127, pp.~31--40,
  2016.

\bibitem{kim2004quadratically}
Y.~Kim and M.~Mesbahi, ``Quadratically constrained attitude control via
  semidefinite programming,'' {\em IEEE Transactions on Automatic Control},
  vol.~49, no.~5, pp.~731--735, 2004.

\bibitem{hutao2011rhc}
C.~Hutao, C.~Xiaojun, X.~Rui, and C.~Pingyuan, ``Rhc-based attitude control of
  spacecraft under geometric constraints,'' {\em Aircraft Engineering and
  Aerospace Technology}, vol.~83, no.~5, pp.~296--305, 2011.

\bibitem{mcinnes1994large}
C.~R. McInnes, ``Large angle slew maneuvers with autonomous sun vector
  avoidance,'' {\em Journal of guidance, control, and dynamics}, vol.~17,
  no.~4, pp.~875--877, 1994.

\bibitem{lee2014feedback}
U.~Lee and M.~Mesbahi, ``Feedback control for spacecraft reorientation under
  attitude constraints via convex potentials,'' {\em IEEE Transactions on
  Aerospace and Electronic Systems}, vol.~50, no.~4, pp.~2578--2592, 2014.

\bibitem{avanzini2009potential}
G.~Avanzini, G.~Radice, and I.~Ali, ``Potential approach for constrained
  autonomous manoeuvres of a spacecraft equipped with a cluster of control
  moment gyroscopes,'' {\em Proceedings of the Institution of Mechanical
  Engineers, Part G: Journal of Aerospace Engineering}, vol.~223, no.~3,
  pp.~285--296, 2009.

\bibitem{wisniewski2005slew}
R.~Wisniewski and P.~Kulczycki, ``Slew maneuver control for spacecraft equipped
  with star camera and reaction wheels,'' {\em Control engineering practice},
  vol.~13, no.~3, pp.~349--356, 2005.

\bibitem{mengali2004spacecraft}
G.~Mengali and A.~A. Quarta, ``Spacecraft control with constrained fast
  reorientation and accurate pointing,'' {\em The Aeronautical Journal},
  vol.~108, no.~1080, pp.~85--91, 2004.

\bibitem{shen2017velocity}
Q.~Shen, C.~Yue, and C.~H. Goh, ``Velocity-free attitude reorientation of a
  flexible spacecraft with attitude constraints,'' {\em Journal of Guidance,
  Control, and Dynamics}, pp.~1--7, 2017.

\bibitem{spindler2002attitude}
K.~Spindler, ``Attitude maneuvers which avoid a forbidden direction,'' {\em
  Journal of dynamical and control systems}, vol.~8, no.~1, pp.~1--22, 2002.

\bibitem{APF2}
Y.~Koren and J.~T. Borenstein, ``Potential field methods and their inherent
  limitations for mobile robot navigation,'' in {\em Proc. of International
  Conference on Robotics and Automation (ICRA)}, vol.~2, pp.~1398--1404, 1991.

\bibitem{okamoto2015novel}
M.~Okamoto and M.~R. Akella, ``Novel potential-function-based control scheme
  for non-holonomic multi-agent systems to prevent the local minimum problem,''
  {\em International Journal of Systems Science}, vol.~46, no.~12,
  pp.~2150--2164, 2015.

\bibitem{doria2013algorithm}
N.~S.~F. Doria, E.~O. Freire, and J.~C. Basilio, ``An algorithm inspired by the
  deterministic annealing approach to avoid local minima in artificial
  potential fields,'' in {\em Advanced Robotics (ICAR), 2013 16th International
  Conference on}, pp.~1--6, IEEE, 2013.

\bibitem{Boada2010}
J.~Boada, C.~Prieur, S.~Tarbouriech, C.~Pittet, and C.~Charbonnel, ``Extended
  model recovery anti-windup for satellite control,'' {\em 18th IFAC Symposium
  on Automatic Control in Aerospace (ACA 2010)}, vol.~43, no.~15, pp.~205--210,
  2010.

\bibitem{burlionIJC2017}
L.~Burlion, J.~M. Biannic, and T.~Ahmed-Ali, ``Attitude tracking control of a
  flexible spacecraft under angular velocity constraints,'' {\em International
  Journal of Control}, vol.~0, no.~0, pp.~1--17, 2017.

\bibitem{garone2016explicit}
E.~Garone and M.~M. Nicotra, ``Explicit reference governor for constrained
  nonlinear systems,'' {\em IEEE Transactions on Automatic Control}, vol.~61,
  no.~5, pp.~1379--1384, 2016.

\bibitem{ERG_CSM}
M.~M. Nicotra and E.~Garone, ``The explicit reference governor,'' {\em Control
  Systems Magazine}, to appear in: August 2018.

\bibitem{SO3_Representation}
M.~D. Shuster, ``A survey of attitude representation,'' {\em The Journal of the
  Astronautical Sciences}, vol.~41, no.~4, pp.~439--517, 1993.

\bibitem{SO3_Control}
N.~A. Chaturvedi, A.~K. Sanyal, and N.~H. McClamroch, ``Rigid-body attitude
  control,'' {\em IEEE Control Systems Magazine}, vol.~31, no.~3, pp.~30--51,
  2011.

\bibitem{khalil}
H.~K. Khalil, {\em Nonlinear Systems}.
\newblock Prentice-Hall, 2002.

\bibitem{ERG_EL}
M.~M. Nicotra and E.~Garone, ``Control of euler-lagrange systems subject to
  constraints: An explicit reference governor approach,'' in {\em Proc. of the
  IEEE Conference on Decision and Control (CDC)}, pp.~1154--1159, 2015.

\bibitem{blanchini}
F.~Blanchini, ``Set invariance in control,'' {\em Automatica}, vol.~35, no.~11,
  pp.~1747--1767, 1999.

\bibitem{APF}
O.~Khatib, ``Real-time obstacle avoidance for manipulators and mobile robots,''
  in {\em Proc. of International Conference on Robotics and Automation (ICRA)},
  vol.~2, pp.~500--505, 1985.

\bibitem{Discontinuous}
J.~Cortes, ``Discontinuous dynamical systems,'' {\em IEEE Control Systems
  Magazine}, vol.~28, no.~3, pp.~36--73, 2008.

\end{thebibliography}


\subsection{Proof of Proposition 2}\label{app:Conservative}
The attraction term \eqref{eq: Attraction Field} is a conservative vector field obtained by computing the gradient of the potential function
\begin{equation}\label{eq: Attraction Potential}
P_r(v,r) = \begin{cases}
\frac2{\sin\frac\eta2}(1\!-\!\cos\frac{\tilde\alpha}2), & \textrm{if}~\tilde{\alpha}\!\in\![0,\eta]; \\
\frac2{\sin\frac\eta2}(1\!-\!\cos\frac\eta2)+\tilde\alpha-\eta, & \textrm{if}~\tilde{\alpha}\!\in\!(\eta,\pi].
\end{cases}
\end{equation}
Likewise, the repulsion term \eqref{eq: Repulsion Field, Exclusion Cone} is a conservative vector field obtained by computing the gradient of the potential function
\begin{equation}\label{eq: Repulsion Potential, Exclusion Cone}
P_e(v) = \begin{cases}
\frac12\frac{(\zeta-\theta(v))^2}{\zeta-\delta}, & \textrm{if}~\theta(v)\!\in\![0,\zeta); \\
0, & \textrm{if}~\theta(v)\!\in\![\zeta,\pi].
\end{cases}
\end{equation}
Given the combined potential function $P(v,r)=P_r(v,r)+P_e(v)$, it follows by construction that the dynamic system \eqref{eq:AttractionField_Conservative} is such that
\[
\dot{P}(v(t),r)=-\|\nabla P(v(t),r)\|^2,
\]
with
\begin{equation}\label{eq: Exclusion Antigradient}
\nabla P(v,r)=-(\rho_r(v,r)+\rho_e(v)).
\end{equation}
By taking advantage of the monotone time-decreasing nature of $P(v(t),r)$, each item in the statement of Proposition \ref{prp: Exclusion} can be proven separately. 

\emph{Point 1.} The invariance of $\mathcal{R}_\delta$ is proven by showing that, whenever $h^T\!R(v)e_1=\cos(\psi_1\!+\!\delta)$, the value of $h^T\!R(v)e_1$ cannot increase further. Following from equation \eqref{eq: Repulsion Field, Exclusion Cone}, the condition $h^T\!R(v)e_1=\cos(\psi_1\!+\!\delta)$ entails $\|\rho_e(v)\|=1$. Since equation \eqref{eq: Attraction Field} ensures $\|\rho_r(v)\|\leq1,~\forall v\in\quat$, it follows from equations \eqref{eq: Repulsion Field, Exclusion Cone} and \eqref{eq: Exclusion Antigradient} that
\begin{equation}
-\nabla P^T\!(v,r)\,\hat{h}R(v)e_1\geq0.
\end{equation}
As a result, the monotone time-decreasing property of $P(v(t),r)$ is sufficient to guarantee $h^T\!R(v(t))e_1\leq\cos(\psi_1\!+\!\delta),~\forall t\geq0$.  

\emph{Point 2.} Due to the requirement $\eta<\delta$, it follows from equation \eqref{eq: Attraction Field} that $\|\rho(s_1,r)\|=1,~\forall r\in\mathcal{R}_\delta$. Likewise, it follows from equation \eqref{eq:SaddleDistance} that $\|\rho_e(s_1)\|=1$. As a result, equation \eqref{eq:SaddleDirection} implies $\rho_r(s_1,r)+\rho_e(s_1)=0,~\forall r\in\mathcal{R}_\delta$. This is sufficient to show that $s_1\in\quat$ satisfying \eqref{eq:SaddleConditions} is an equilibrium point.

To study the stability properties of $s_1$, consider the case $v\in\mathcal{S}_1$. Following from \eqref{eq:SaddleManifold}, any point belonging to the manifold $\mathcal{S}_1$ is such that $\rho_r(v,r)+\rho_e(v)$ is parallel to $\rho_r(v,r)$. Therefore, system \eqref{eq:AttractionField_Conservative} will remain in the manifold and will necessarily converge to $s_1\in\mathcal{S}_1$ due to the time-decreasing properties of $P(v(t))$ and the fact that $r\notin\mathcal{S}_1$.

Consider now the case $v=ws_1$, where the quaternion $w\in\quat$ satisfying $w_R=\cos\frac\varepsilon2$, $w_I\perp\hat{h}R(s_1)e_1$ represents an infinitesimal rotation of $\varepsilon>0$ away from the manifold $\mathcal{S}_1$. Following from \eqref{eq: Attraction Potential}, \eqref{eq:SaddleConditions}, and the condition $\eta<\delta$,
\begin{equation}\label{eq:Variation_Pr}
P_r(ws_1,r)-P_r(s_1,r)=\Phi(\arg(s_1r^*)),
\end{equation}
where $\arg(q):=2\arccos(q_R)$ and 
\begin{equation}\label{eq:Phi}
\Phi(x):=2\arccos\left(\cos\left(\frac{\varepsilon}2\right)\cos\left(\frac x2\right)\right)-x
\end{equation}
is a strictly monotone decreasing function which is equivalent to $\Phi(\arg(s_1r^*))=\arg(ws_1r^*)-\arg(s_1r^*)$. Likewise, it follows from \eqref{eq: Repulsion Potential, Exclusion Cone} that
\begin{equation}\label{eq:Variation_Pe}
P_e(ws_1)-P_e(s_1)=-\Phi(\vartheta(s_1))+\frac{\Phi(\vartheta(s_1)))^2}{2(\zeta-\delta)},
\end{equation}
where $\vartheta(s_1):=\arccos(h^T\!R(s_1)e_1)-\psi_1$. By combining equations \eqref{eq:Variation_Pr}, \eqref{eq:Variation_Pe} it follows that
\begin{equation}
P(ws_1,r)-P(s_1,r)\!=\!\Phi(\arg(s_1r^*))\!-\!\Phi(\vartheta(s_1))\!+\!O(\varepsilon^2).
\end{equation}
Due to equation \eqref{eq:SaddleConditions}, $r\in\mathcal{R}_\zeta$ implies $\arg(s_1r^*)>\vartheta(s_1)$, thus leading to $P(ws_1,r)<P(s_1,r)$ for an arbitrarily small perturbation $\varepsilon>0$. This is sufficient to show that the equilibrium point $s_1$ is unstable in every direction perpendicular to $\hat{h}R(s_1)e_1$. As a result, for all $v\not\in\mathcal{S}_1$, system \eqref{eq:AttractionField_Conservative} cannot converge to $s_1$.

\emph{Point 3.} The asymptotic stability of the equilibrium point $v=r$ follows directly from the fact that $P(r,r)=0$, $\forall r\in\mathcal{R}_\zeta$, the function $P(v,r)$ is positive definite, and the trajectory of \eqref{eq:AttractionField_Conservative} is such that $P(v(t),r)$ is monotonically time-decreasing. The attraction basin can thus be estimated by taking the entirety of $\mathcal{R}_\delta$ and subtracting the attraction basin of the saddle point $s_1$.

\subsection{Proof of Proposition 3}\label{app:Global}
Given the potential function $P(v,r)=P_r(v,r)+P_e(v)$ obtained using \eqref{eq: Attraction Potential}, \eqref{eq: Repulsion Potential, Exclusion Cone}, it follows by construction that the dynamic system \eqref{eq:AttractionField_Global} is such that
\[
\dot{P}(v(t),r)=-\nabla P(v(t),r)^T\,\bigl(\nabla P(v(t),r)\!+\!\rho_d(v(t),r)\bigr),
\]
with $\nabla P(v,r)$ given in \eqref{eq: Exclusion Antigradient}. Due to equations \eqref{eq: Destabilization Field}-\eqref{eq:DestabilizationDirection}, it follows that $\nabla P(v(t),r)^T\rho_d(v(t),r)=0$, thus implying that the introduction of the saddle destabilization term does not compromise the monotone time-decreasing properties of $P(v(t),r)$ detailed in the proof of Proposition \ref{prp: Exclusion}. Each item in the statement will thus be proven separately. 

\emph{Point 1.} The invariance of $\mathcal{R}_\delta$ is proven using the same arguments as Proposition \ref{prp: Exclusion}. In particular, it is sufficient to note that $\rho_d(v,r)$ cannot cause a violation of constraints due to the fact that $\rho_d(v,r)\hat{h}R(v)e_1=0$ whenever $h^T\!R(v)e_1=\cos(\psi_1\!+\!\delta)$.  

\emph{Point 2.} Since $\rho_d(v,r)$ is perpendicular to $\rho_r(v,r)+\rho_e(v)$, the destabilization term is unable to generate additional equilibria with respect to the ones identified in the proof of Proposition \ref{prp: Exclusion}. Moreover, it follows from \eqref{eq: Destabilization Field} that $\rho_d(s_1,r)\neq0$. This implies that $s_1$, i.e. the saddle point of the potential function $P(v)$, is no longer equilibrium points for the system \eqref{eq:AttractionField_Global}, which now admits $r$ as the only equilibrium point.

Since the vector field \eqref{eq:AttractionField_Global} is not Lipschitz, however, it is also necessary to show that the discontinuity found in \eqref{eq: Destabilization Field}-\eqref{eq:DestabilizationDirection} cannot cause limit cycles. To do so, let $d$ be a generic point in the influence region
\begin{equation}\label{eq:InfluenceRegion}
\mathcal{D}_1:=\left\{v\!\in\!\mathcal{S}_1~\left|~
h^T\!R(v)e_1\geq\cos(\psi_1+\zeta)\right.\right\},
\end{equation}
and let $w\in\quat$ satisfying $w_R=\cos\frac\varepsilon2$, $w_I\perp\hat{h}R(s_1)e_1$ be an infinitesimal rotation of $\varepsilon>0$ away from the domain $\mathcal{D}_1$. Then, the value of the vector field in $wd$ is $\rho_r(wd,r)+\rho_e(wd)+\rho_d(wd,r)$. As already proven in Point 2 of Proposition \ref{prp: Exclusion}, the vector field $\rho_r(wd,r)+\rho_e(wd)$ points away from the set $\mathcal{D}_1$. As for $\rho_d(wd,r)$, it follows from \eqref{eq:DestabilizationDirection} that $\rho_d(wd,r)^Tw=0$, meaning that the destabilization terms does not point towards the set $\mathcal{D}_1$. Since the set $\mathcal{D}_1$ is repulsive, it follows from \cite{Discontinuous} that there are no limit cycles associated to this discontinuity, therefore the point $v=r$ is the only asymptotic solution of the proposed flow field.

\subsection{Proof of Corollary 1}\label{app:Corollary}
The proof is analogous to the one provided for Proposition \ref{prp: Destabilization}. Indeed, the potential function $P(v,r)=P_r(v,r)+P_e(v)$ is monotone time decreasing and is characterized by $l$ saddle points, each of which satisfy equations \eqref{eq:SaddleConditions} for a different exclusion cone constraint. Given the influence regions \eqref{eq:InfluenceRegion} associated to each exclusion cone constraint, it follows from Assumption \ref{ass:DistanceMargin} that $\mathcal{D}_i\cap\mathcal{D}_j=\emptyset$, $\forall i\neq j$. As a result, the local properties of each saddle point $s_i$ are not influenced by the presence of the other exclusion cone constraints. Global Asymptotic Stability therefore follows from the fact that $r$ is the only equilibrium point in the domain $\mathcal{R}_\delta$, there are no limit cycles due to the discontinuities, and $P(v(t),r)$ is a time-decreasing function.

\end{document}